\documentclass[tran,twocolumn,noversion,nonote]{cdcarticle}

\def\MODE{1}

\usepackage{cite}
\usepackage{graphicx}
\usepackage{subcaption}
\usepackage[cmex10]{amsmath}
\usepackage{amsthm,amssymb}
\usepackage{array}
\usepackage{calc}
\usepackage{enumerate}
\usepackage[hidelinks]{hyperref}
\usepackage{arydshln}  % dashed lines in arrays
\usepackage{combelow}

\usepackage{tikz}
\usetikzlibrary{calc}

\usepackage{pgfplots}
\tikzset{
very thick,>=latex,
block/.style={auto,draw,rectangle,minimum height=1.8em,minimum width=2em},
bline/.style={blue,very thick},
bdash/.style={blue,very thick,dashed},
blfill/.style={blue,fill=blue!20,very thick},
bfill/.style={fill=blue!20},
twopic/.style={scale=2},
threepic/.style={scale=1.3},
fourpic/.style={scale=1.6},
axes/.style={black,very thick,<->},
ticks/.style={black,very thick}
}

\makeatletter
\def\@IEEElegacywarn#1#2{}
\makeatother

\newtheorem{thm}{Theorem}

\newtheorem{defn}[thm]{Definition}
\newtheorem{cor}[thm]{Corollary}

\renewenvironment{proof}{\noindent{\bf Proof.}}{ \hfill ~\qed}
\def\qed{\rule[0pt]{5pt}{5pt}\par\medskip}

\newcommand{\T}{\rule{0pt}{2.6ex}}
\newcommand{\hlinet}{\hline\T}

\newcommand{\bmat}[1]{\begin{bmatrix}#1\end{bmatrix}}

\newcommand{\eps}{\varepsilon}
\newcommand{\R}{\mathbb{R}} % Real numbers
\newcommand{\Z}{\mathbb{Z}} % Real numbers
\newcommand{\C}{\mathbb{C}} % Complex numbers
\newcommand{\Htwo}{\mathcal{H}_2}

\renewcommand{\L}{\mathcal{L}}
\newcommand{\LTI}{\L_\textup{TI}}
\renewcommand{\t}[1]{\ensuremath{\tilde{#1}}}

\newcommand{\eemph}[1]{\textbf{\textit{#1}}}
\renewcommand{\subset}{\subseteq}

\newcommand{\norm}[1]{\lVert{#1}\rVert}

\newcommand{\set}[2]{\ensuremath{\left.\left\{ #1 \,\,\right\vert\, #2 \right\}}}
\newcommand{\sett}[2]{\ensuremath{\left\{ #1 \,\left\vert\, #2 \right\}\right.}}

\def\hg{h}
\def\ng{N}
\def\mg{M}

\def\note#1{}

\begin{document}
\title{Convexity of Decentralized Controller Synthesis}

\if\MODE1
	\author{Laurent~Lessard~\and~Sanjay~Lall}
\else
	\author{Laurent~Lessard~and~Sanjay~Lall
	\thanks{L. Lessard completed this work while he was at the Department of Aeronautics and Astronautics at Stanford University, Stanford, CA 94305, USA. He is now with the Department of Mechanical Engineering at the University of California, Berkeley, CA 94720, USA. \texttt{lessard@berkeley.edu}}%
	\thanks{S. Lall is with the Department of Electrical Engineering and Aeronautics and Astronautics at Stanford University, Stanford, CA 94305, USA. \texttt{lall@stanford.edu}}}
\fi

% only for cdcarticle
\note{Submitted to IEEE Transactions on Automatic Control}

\ifCLASSOPTIONpeerreview
\markboth{IEEE Transactions on Automatic Control}{}
\else
\markboth{IEEE Transactions on Automatic Control}%
{Lessard \MakeLowercase{\textit{et al.}}: An Algebraic Characterization of Tractability for Decentralized Control}
\fi

\maketitle

%%%%%%%%%%%%%%%%%%%%%%%%%%%%%%%%%%%%%%%%%%%%%%%%%%%%%%%%%%%%%%%%%%%%%%%%%%%%%%
\begin{abstract}
In decentralized control problems, a standard approach is to specify
the set of allowable decentralized controllers as a closed subspace of
linear operators. This then induces a corresponding set of Youla
parameters. Previous work has shown that quadratic invariance of the
controller set implies that the set of Youla parameters is convex. In
this paper, we prove the converse. We thereby show that the only
decentralized control problems for which the set of Youla parameters
is convex are those which are quadratically invariant. We further show
that under additional assumptions, quadratic invariance is necessary
and sufficient for the set of achievable closed-loop maps to be
convex. We give two versions of our results. The first applies to
bounded linear operators on a Banach space and the second applies to
(possibly unstable) causal LTI systems in discrete or continuous time.
\end{abstract}

\IEEEpeerreviewmaketitle

%%%%%%%%%%%%%%%%%%%%%%%%%%%%%%%%%%%%%%%%%%%%%%%%%%%%%%%%%%%%%%%%%%%%%%%%%%%%%%%
\section{Introduction}\label{sec:intro}

\IEEEPARstart{T}{his} paper
 considers the feedback control of linear systems subject to
 structural constraints on the controller. We are interested in
 characterizing when the set of achievable closed-loop maps is
 convex. Convexity is important because in many cases it makes the
 problem of synthesizing structurally constrained controllers that
 optimize some performance objective a tractable one.

Suppose $w\in \mathcal{W}$ is the exogenous disturbance,
$u\in\mathcal{U}$ is the control input, $z\in\mathcal{Z}$ is the
regulated output, and $y\in\mathcal{Y}$ is the measurement.  We assume
that the plant is a linear and continuous map
$P:(\mathcal{W,U})\to(\mathcal{Z,Y})$. The controller is a map
$K:\mathcal{Y}\to\mathcal{U}$ that is connected to the plant in
feedback. We partition the four blocks of $P$ as $P_{ij}$ and by
convention we let $G = P_{22}$. The closed-loop interconnection
mapping $w\mapsto z$ is depicted in Figure~\ref{fig:PK_diagram}.

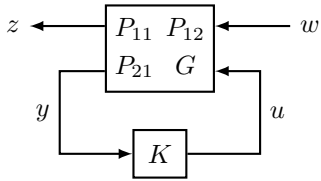
\begin{figure}[ht]
\centering
\begin{tikzpicture}[thick,node distance=1.4cm]
\node [block](P){$\begin{matrix}P_{11}\rule[-1.3ex]{-1.3ex}{0pt}& P_{12}\\
P_{21}\rule{-1.3ex}{0pt} & G\end{matrix}$};
\node [block,below of=P](K){$K$};
\draw [<-] (P.east)+(0,-0.3) -- +(0.6,-0.3) |- node[pos=0.25,anchor=west]{$u$} (K);
\draw [<-] (P.east)+(0,0.3) -- +(1,0.3) node [anchor=west]{$w$};
\draw [->] (P.west)+(0,-0.3) -- +(-0.6,-0.3) |- node[pos=0.25,anchor=east]{$y$} (K);
\draw [->] (P.west)+(0,0.3) -- +(-1,0.3) node [anchor=east]{$z$};
\end{tikzpicture}
\caption{\label{fig:PK_diagram}Plant $P$ connected to a controller $K$.}
\end{figure}

The set of \emph{achievable} closed-loop
maps is
\[
C = \set{ P_{11} + P_{12}K(I-GK)^{-1}P_{21} }{ K \in S }
\]
Here the set $S$ is the set of $K$ satisfying the structural constraints
imposed in the problem. Examples of such structural constraints
include sparsity requirements, where for example the controller
responsible for choosing~$u_j$ cannot measure~$y_i$. Other
possibilities include constraints arising as a consequence of
measurement delays.  To make notation more compact, we define the
function
\begin{equation}\label{eq:h}
\hg(K) = -K(I-GK)^{-1}
\end{equation}
and then we may write $C = P_{11}-P_{12}\hg(S)P_{21}$ where $\hg(S)$
is the image of $S$ under the map $\hg$.  Note that $\hg$ implicitly
depends on the map $G$.  In this paper we refer to $\hg$ and $G$
together with the understanding that $\hg$ is defined as a function of
$G$.  A constrained synthesis problem may take the form
\begin{equation}\label{opt:synth}
	\begin{aligned}
	\text{minimize}&\qquad \norm{X} \\
	\text{subject to}&\qquad X \in C
	\end{aligned}
\end{equation}
where $\norm\cdot$ is some convex measure of performance, such as a
norm. It is desirable to know when $C$ is a convex set, because then
if~\eqref{opt:synth} is a convex optimization problem.

A set $S$ is called \emph{quadratically invariant} under $G$ if
$KGK\in S$ for all $K\in S$.  It is shown
in~\cite{rotkowitz02,rotkowitz06} in that, roughly speaking,
$\hg(S)=S$ if and only if $S$ is quadratically invariant under
$G$. This provides a sufficient condition under which $C$ is convex.

This work builds on the notion of quadratic invariance developed
in~\cite{rotkowitz02,rotkowitz06}.  This earlier work showed that
quadratic invariance is sufficient for $\hg(S)$ and $C$ to be convex.
In this paper, our main contributions are as follows.

\begin{enumerate}[(i)]

\item We show that quadratic invariance is a \emph{necessary} condition for convexity of $\hg(S)$.

\item We show that, subject to some additional assumptions, quadratic invariance is also a necessary
condition for convexity of~$C$.

\item We give examples that illustrate both the generality of the results
 as well as the limitation encountered when the additional assumptions
 in (ii) are violated.

\end{enumerate}

The paper is organized as follows. We begin by reviewing related works
in the literature and covering necessary mathematical preliminaries,
including a review of quadratic invariance. In Section~\ref{sec:main}
we prove our main results and in Section~\ref{sec:examples} we show
some illustrative examples. Finally, we conclude in
Section~\ref{sec:conclusion}.

\subsection{Prior work}

Optimal controller synthesis subject to information constraints is
known to be hard in general~\cite{blondel}. Even when the plant is
linear, the noise is Gaussian, and the cost function is quadratic, the
optimal controller may not be linear~\cite{witsenhausen}. Furthermore,
there is currently no known efficient algorithm for finding the
optimal linear controller in general.

Despite these difficulties, the pervasiveness of decentralized
information in large-scale systems has driven researchers to seek
subclasses of problems that are tractable. Early work by
Radner~\cite{radner} showed that some static team decision problems
admit optimal controllers that are linear. This was extended to
dynamic teams by Ho and Chu~\cite{hochu} under the assumption that the
information structure is~\emph{partially nested}.

More recently, many works have addressed the broad
area of decentralized control synthesis, including~\cite{lin_2011,nayyar_2013,sabau_2011,tanaka_2011}.
Specific  efforts have focused on characterizing when the set of
achievable closed-loop maps as detailed in Section~\ref{sec:intro} is
convex. For example,~\cite{qimurti04,nonclassical} shows several
classes of information constraints that lead to a convex~$C$.
In~\cite{rotkowitz06} it is shown that the simple algebraic condition
of quadratic invariance encompasses a wide class of problems for
which~$C$ is convex.

However, quadratic invariance is only a sufficient condition for
convexity of~$C$. Recent
work~\cite{lessard2009reduction,lessard2010iqi} introduces the concept
of~\emph{internal quadratic invariance}, which gives a more complete
characterization of~$C$. This condition highlights cases where
quadratic invariance does not hold, but a suitable transformation can
produce a new problem for which~$C$ is unchanged and quadratic
invariance now holds.

One may also characterize~$C$ more directly. Using tools from
algebraic geometry, specifically elimination theory, 
the paper~\cite{shin2012decentralized} gives a method for
computing a representation of the smallest algebraic variety
containing~$C$. This does not necessarily
make the problem of verifying convexity easier, but it does open the
door to other tools for identifying convexity such as sum-of-squares
relaxations. See~\cite{sos-convexity} and references therein.

This paper gives conditions under which quadratic invariance is both
necessary and sufficient for convexity. Preliminary versions of some
of these results appeared
in~\cite{lessard_thesis,lessard2011qi_necc_suff} for the Banach space
case only. In this paper, we include new results that cover the
extended spaces~$\ell_{2e}$ and~$L_{2e}$, and we present new examples
in Section~\ref{sec:examples} that illustrate both the generality and
limitations of our results.

%\newpage
%%%%%%%%%%%%%%%%%%%%%%%%%%%%%%%%%%%%%%%%%%%%%%%%%%%%%%%%%%%%%%%%%%%%%%%%%%%%%%%
\section{Preliminaries}\label{sec:prelim}

If $\mathcal{X}$ and $\mathcal{Y}$ are topological vector spaces
(TVS), we let $\L(\mathcal{X},\mathcal{Y})$ denote the set of all maps
$T:\mathcal{X}\to\mathcal{Y}$ such that $T$ is linear and
continuous. This paper concerns properties of linear operators. What
we are able to prove depends on the underlying structure of the vector
spaces involved, so we distinguish two classes of topological vector
spaces: Banach spaces, and the extended spaces $\ell_{2e}$ and
$L_{2e}$.

\subsubsection*{Banach spaces}
A Banach space is a TVS whose topology is induced by a norm, and the
space is complete. Any linear and continuous map from one Banach space
to another is bounded. Common examples of Banach spaces include
$\ell_p$ and $L_p$; the set of functions $f:\Z_+\to \R$ and
Lebesgue-measurable functions $f:\R_+\to\R$ respectively, for which
the $p$-norm is finite. Simpler examples include the Hilbert spaces
$\R^n$, $\ell_2$, or $L_2$. Suppose $\mathcal{U}$ and $\mathcal{Y}$
are Banach spaces, and $G \in \L(\mathcal{U},\mathcal{Y})$. Define the
following set.
\begin{equation}\label{def:M}
\mg = \set{ K \in \L(\mathcal{Y},\mathcal{U}) }{ I-GK \text{ is invertible} }
\end{equation}
Note that the set $\mg$ is precisely the domain of $\hg$. For any
$A \in \L(\mathcal{Y},\mathcal{Y})$, the \eemph{resolvent set} is
defined as $
\rho(A) = \set{ \lambda \in \C }{ (\lambda I - A) \text{ is invertible}} 
$.  It is a fact that $\rho(A)$ is always an open set, and contains
all $\lambda \in \C$ for which $|\lambda| > \norm{A}$. Therefore, one
may define $\rho_\text{uc}(A)$, the unbounded connected component of
$\rho(A)$. Now define the subset $\ng\subset \mg$ as follows.
\begin{equation}\label{def:N}
\ng = \set{K \in \L(\mathcal{Y},\mathcal{U}) }{ 1 \in \rho_\text{uc}(GK) }
\end{equation}

\subsubsection*{Extended spaces}

Banach spaces are well-suited for representing a wide variety of
systems, but when the time horizon is infinite, such as with $\ell_2$
and $L_2$, only bounded maps are permitted. In order to represent
unbounded maps as well, such as unstable systems, we use the notion
of \emph{extended spaces}. First, define the truncation and shift
operators $P_T$ and $D_T$ which operate on functions $f:\Z_+\to\R$
or $f:\R_+\to\R$, as follows. For $T >0$, define
\begin{align*}
P_T f = \begin{cases}
	f(t)	& \text{if } t \le T \\
	0		& \text{otherwise}
\end{cases}
\end{align*}
and
\begin{align*}
D_T f = \begin{cases}
	0			& \text{if } t < T \\
	f(t-T)	& \text{otherwise}
\end{cases}
\end{align*}
The extended spaces $\ell_{pe}$ and $L_{pe}$ are defined as follows.
\begin{align*}
\ell_{pe} &= \sett{f:\Z_+\to\R}{P_Tf\in\ell_p\text{ for all }T\in\Z_+} \\
L_{pe} &= \sett{f:\R_+\to\R}{P_Tf\in L_p\text{ for all }T\in\R_+}
\end{align*}
Note that $\ell_{pe}$ is the same for every $p$, and is the set of
real sequences, which we abbreviate as $\ell_e$. However, $L_{pe}$ is
the set of functions that are $L_p$ on finite intervals, and so is
different for different~$p$. 

A map $A\in\L(L_{2e}^m,L_{2e}^n)$ or $A\in\L(\ell_{2e}^m,\ell_{2e}^n)$ is called causal if,
for all $ T \geq 0 $, we have $P_T A P_T = P_T A $.
A causal system $A$  is called time-invariant if  $D_T A = A D_T $ for all $T >0$.
We denote the set of linear, causal, time-invariant maps  as $\LTI$.

The extended spaces are not Banach spaces, but we may endow them with
a suitable topology and recover notions of convergence and
continuity. Let the topology on $L_{2e}$ be generated by all the open
$\norm{\cdot}_T$-balls for every $T\in\R_+$, where we define
$\norm{f}_T = \norm{P_T f}_{L_2}$. Also, we let the topology on
$\L(L_{2e}^m,L_{2e}^n)$ be generated by all the open
$\norm{\cdot}_T$-balls for every $T\in\R_+$, where $\norm{A}_T
= \norm{P_TA}_{L_2^m\to L_2^n}$. It can be shown that these topologies
are Hausdorff, and thus $L_{2e}$ and $\L(L_{2e}^m,L_{2e}^n)$ are
topological vector spaces (TVS). Furthermore, convergence in this
topology is equivalent to convergence in every $\norm{\cdot}_T$-norm,
and continuity of a linear operator in this topology is equivalent to
continuity in every $\norm{\cdot}_T$-norm. A similar topology is
defined for $\ell_{2e}$. For a thorough treatment of these concepts,
see for example~\cite{zimmer}. We also require the concept of
an \emph{inert} subspace, which we give below.

%%%%%%%%%%%%%%%%%%%%%%%%%%%%%%%%%%%%%%%%%%%%%%%%%%%%%%%%%%%%%%%%%%%%%%%%%%%%%%

\begin{defn}\label{def:inert}
The set $S \subset \LTI(L_{2e}^{n_u},L_{2e}^{n_y})$ is \eemph{inert}
with respect to $G \in \LTI(L_{2e}^{n_y},L_{2e}^{n_u})$ if for all
$K\in S$, $(gk)_{ij} \in L_{\infty e}$ for all $i,j=1,\dots,n_y$ where
$(gk)$ is the impulse response matrix of $GK$. We overload our
notation and also define
$S\subset \LTI(\ell_{2e}^{n_u},\ell_{2e}^{n_y})$ to be inert with
respect to $G \in \LTI(\ell_{2e}^{n_u},\ell_{2e}^{n_y})$ if for all
$K\in S$, $(gk)_{ij}\in\ell_e$ for all $i,j=1,\dots,n_y$ and
$r((gk)(0)) < 1$ where $(gk)$ is the discrete impulse response matrix
of $GK$ and $r(\cdot)$ denotes the spectral radius.
\end{defn}

Note in particular that if $S$ is inert with respect to $G$, then
$I-GK$ is invertible for all $K\in S$. A proof of this result is
in~\cite{rotkowitz06}.

\subsubsection*{Quadratic invariance} 

We now summarize the definitions and main results regarding quadratic
invariance. The following definition may be found
in~\cite{rotkowitz02,rotkowitz06}.

\begin{defn}\label{def:qi}
Suppose $\mathcal{U}$ and $\mathcal{Y}$ are TVS. Suppose $S \subset \L(\mathcal{Y},\mathcal{U})$ and $G\in\L(\mathcal{U},\mathcal{Y})$.
The set $S$ is \eemph{quadratically invariant} under $G$ if $KGK \in S$ for all $K \in S$.
\end{defn}

\noindent The first quadratic invariance result we present applies to Banach spaces.

\begin{thm}[from \cite{rotkowitz02}] \label{thm:qi_banach} Suppose $\mathcal{U}$ and $\mathcal{Y}$ are Banach spaces. Suppose $S \subset \L(\mathcal{Y},\mathcal{U})$ is a closed subspace and $G \in \L(\mathcal{U},\mathcal{Y})$. Define $\mg$ and $\ng$ as in~\eqref{def:M}--\eqref{def:N}, and suppose $S \cap \ng = S \cap \mg$. Then
$S$ is quadratically invariant under $G$ if and only if
$ \hg(S\cap \mg) = S \cap \mg$.
\end{thm}

\noindent 
The second quadratic invariance result applies to the topological
vector spaces~$\ell_{2e}$ and~$L_{2e}$.

\begin{thm}[from \cite{rotkowitz06}] \label{thm:qi_extended} Suppose $S\subset\LTI(L_{2e}^{n_u},L_{2e}^{n_y})$ is an inert closed subspace and $G \in \LTI(L_{2e}^{n_u},L_{2e}^{n_y})$, or suppose $S\subset\LTI(\ell_{2e}^{n_u},\ell_{2e}^{n_y})$ is an inert closed subspace and $G \in \LTI(\ell_{2e}^{n_u},\ell_{2e}^{n_y})$. Then
$S$ is quadratically invariant under $G$ if and only if $\hg(S) = S$.
\end{thm}

Roughly speaking, $S$ being quadratically invariant under $G$ means
that the set $S$ is invariant under the application of
$\hg$. Technical conditions arise only to ensure that we avoid $K$ for
which $\hg(K)$ is not well-defined. In the case of
Theorem~\ref{thm:qi_banach}, this amounts to intersecting $S$ with
$\mg$, the domain of $\hg$. In Theorem~\ref{thm:qi_extended}, we
make the technical assumption that $S$ is an inert subspace, and this
ensures that $\hg(K)$ is always well-defined.

%\newpage

%%%%%%%%%%%%%%%%%%%%%%%%%%%%%%%%%%%%%%%%%%%%%%%%%%%%%%%%%%%%%%%%%%%%%%%%%%%%%%%
\section{Main results}\label{sec:main}

We begin by defining two types of sets that will be useful in our main
results. These definitions apply to any vector space over $\R$, so in
particular they apply to both Banach spaces and the topological vector
spaces $\ell_{2e}$ and $L_{2e}$.

\begin{defn}
Suppose $\mathcal{V}$ is a vector space over~$\R$, and
$S \subset \mathcal{V}$. We call $S$ \eemph{homogeneous} if for all
$v \in S$ and $\alpha \in \R$, we have $\alpha v \in S$.
\end{defn}

\noindent
Homogeneous sets are collections of lines that pass through the
origin. If a point belongs to a homogeneous set, so does the line that
passes through that point and the origin. Note that every subspace is
homogeneous, but not all homogeneous sets are subspaces.

\begin{defn}
Suppose $\mathcal{V}$ is a vector space over~$\R$, and
$T \subset \mathcal{V}$. We call $T$ \eemph{star-shaped} if for all
$v \in T$ and $\alpha \in [0,1]$, we have $\alpha v \in T$.
\end{defn}

\noindent 
If a point belongs to a homogeneous set, so does the line segment
connecting it to the origin. Therefore, every homogeneous set is
star-shaped, but not vice-versa. Furthermore, every convex set containing
the origin is star-shaped, but not vice-versa. For examples of these
sets in $\R^2$, see Figure~\ref{fig:linestar_sets}.

\begin{figure}[ht]
\centering
\begin{minipage}[b]{.45\textwidth}
	\centering
	\begin{tikzpicture}[bline,twopic]
	\def\w{.3}  \def\x{.7}  \def\y{.8}  \def\z{.3}
	\draw[white,fill=blue!20] (-\x,-1) -- (\x,1) -- (1,\x) -- (-1,-\x);
	\draw[<->] (-\x,-1) -- (\x,1);
	\draw[<->] (-1,-\x) -- (1,\x);
	\draw[white,fill=blue!20] (-\y,1) -- (\y,-1) -- (1,-\z) -- (-1,\z);
	\draw[<->] (-\y,1) -- (\y,-1);
	\draw[<->] (1,-\z) -- (-1,\z);
	\draw[<->] (1,\w) -- (-1,-\w);
	\draw[axes] (-1.3,0) -- (1.3,0);
	\draw[axes] (0,-1.1) -- (0,1.1);
	\end{tikzpicture}
	\subcaption{homogeneous set}%\label{fig:lineset}
\end{minipage}
\begin{minipage}[b]{.45\textwidth}
	\centering
	\begin{tikzpicture}[bline,twopic]
	\def\w{.5}  \def\x{.5}
	\draw[white,fill=blue!20] (0,0) -- (-\x,-1) -- (-1,-\x);
	\draw[->] (0,0) -- (-\x,-1);
	\draw[->] (0,0) -- (-1,-\x);
	\draw[->] (0,0) -- (\w,-1);
	\draw (0,0) -- (-.7,.4);
	\draw (0,0) -- (-.7,.9);
	\draw (0,0) -- (-.25,1);
	\draw[white,fill=blue!20] (0,0) -- (.3,.8) -- (.45,.35) -- (1,.2) -- (0,0);
	\draw[bline] (0,0) -- (.3,.8) -- (.45,.35) -- (1,.2) -- (0,0);
	\draw[axes] (-1.3,0) -- (1.3,0);
	\draw[axes] (0,-1.1) -- (0,1.1);
	\end{tikzpicture}
	\subcaption{star-shaped set}%\label{fig:starset}
\end{minipage}
\caption{\label{fig:linestar_sets}
Example of a homogeneous set and a star-shaped set in $\R^2$. Lines
with arrows indicate that they extend to infinity. }
\end{figure}
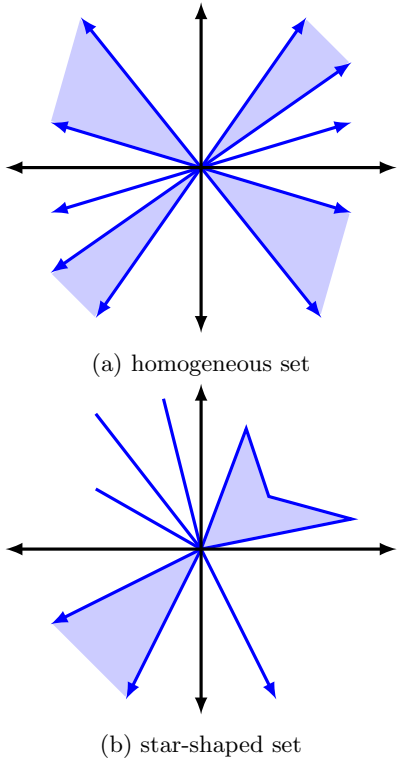

%\newpage
\subsection{Banach space results}\label{sec:main_banach}

Our first main result states that if $\hg$ maps a homogeneous subset
of its domain to a star-shaped set, it must in fact map that
homogeneous set to itself. We will see that this result has
implications concerning convexity.

\begin{thm}\label{thm:main_banach}
Suppose $\mathcal{U}$ and $\mathcal{Y}$ are Banach spaces. Suppose
$S \subset \L(\mathcal{Y},\mathcal{U})$ is closed and
homogeneous and $G \in \L(\mathcal{U},\mathcal{Y})$.
 Define $\mg$ and $\ng$ as
in~\eqref{def:M}--\eqref{def:N}, and suppose $S \cap \ng = S \cap
\mg$. If $\hg(S \cap \mg) = T \cap \mg$ for some star-shaped set $T$, then $T\cap \mg = S\cap \mg$.
\end{thm}

\begin{proof}
Suppose $\hg(S\cap \mg) = T\cap \mg$, where $T$ is a star-shaped
set. Fix some $K\in S\cap \mg$. Therefore, $I-GK$ is invertible, and
$1\in \rho(GK)$. The resolvent set of a bounded linear operator is an
open set, so there exists a sufficiently small $\eps \in (0,1]$ such
that $1-\alpha \in \rho(GK)$ for all $\alpha \in [0,\eps]$. For any
such $\alpha$, it follows that $I-(1-\alpha)GK$ is
invertible. Therefore, $\left(I-(1-\alpha)GK\right)(I-GK)^{-1}$ is
invertible as well. Expanding this expression, we find that it equals
$I-\alpha G\hg(K)$. Thus $\alpha \hg(K) \in \mg$.

Also, we have $\hg(K) \in T$ by assumption, and so $\alpha \hg(K) \in
T$ whenever $\alpha\in[0,1]$, because $T$ is a star-shaped set. It
follows that for $\alpha~\in [0,\eps]$, $\alpha \hg(K) \in
T\cap \mg$. Applying $\hg$ to both sides, we conclude that
$\hg(\alpha \hg(K)) \in \hg(T\cap \mg) = S\cap \mg$, where we made use
of the involutive property of $\hg$. Expanding $\hg(\alpha \hg(K))$,
we find that it equals $\alpha K(I-(1-\alpha)GK)^{-1}$. Since $S$
is a homogeneous set, we may multiply this expression by $-1/\alpha$
when $\alpha\ne 0$, and the result will still lie in $S$. Thus,
$-K(I-(1-\alpha)GK)^{-1} \in S$.  \\Now define the function
$g:[0,\eps] \rightarrow \mathcal{L}(\mathcal{Y},\mathcal{U})$ by
\begin{equation}\label{def:g}
g(\alpha) = -K(I-(1-\alpha)GK)^{-1}
\end{equation}
Notice that $g$ is continuous at $0$, since $I-(1-\alpha)GK$ is
invertible for sufficiently small $\alpha \geq 0$ as above, and the inversion
map is continuous on its domain.  Since $S$ is closed 
and $g(\alpha)\in S$ for $\alpha \in (0,\eps]$, we
have
\[
\lim_{\alpha \rightarrow 0^+} g(\alpha) \in S
\]
We may take the limit $\alpha\to 0^+$ by simply evaluating $g$ at
$\alpha=0$. Thus, we conclude that $\hg(K) \in S$. Now $\hg$ is a
bijection from $\mg$ to $\mg$, and so we actually have $\hg(K) \in
S\cap \mg$. Since $K$ was an arbitrary element of $S\cap \mg$, it
follows that $\hg(S\cap \mg) \subset S\cap \mg$. Using the involutive
property of $\hg$ once more, $\hg(S\cap \mg) = S\cap \mg$, as
required.
\end{proof}

Theorem~\ref{thm:main_banach} may be specialized to the case where $S$
is a subspace, and combined with Theorem~\ref{thm:qi_banach} to yield
a necessary condition for convexity of $\hg$.

\begin{cor}\label{cor:main_banach}
Suppose $\mathcal{U}$ and $\mathcal{Y}$ are Banach spaces. Suppose
$S \subset \L(\mathcal{Y},\mathcal{U})$ is a closed subspace and
$G \in \L(\mathcal{U},\mathcal{Y})$. Define $\mg$ and $\ng$ as
in~\eqref{def:M}--\eqref{def:N}, and suppose $S \cap \ng =
S \cap \mg$. Then the following statements are equivalent.
\begin{enumerate}[(i)]
\item $S$ is quadratically invariant under $G$
\item $\hg(S\cap \mg)=S\cap \mg$
\item $\hg(S\cap \mg)= \Gamma \cap \mg$ for some convex set  $\Gamma$
\end{enumerate}
\end{cor}

\begin{proof}
{\it(i)$\iff$(ii)} is precisely Theorem~\ref{thm:qi_banach}.  The case
{\it(ii)$\implies$(iii)} is immediate.  Finally, to show
{\it(ii)$\impliedby$(iii)}, notice that if \textit{(iii)} holds then
$\Gamma$ must contain the origin, and since it is convex it must
therefore be star-shaped. Then, since every subspace is homogeneous,
\textit{(ii)} follows from Theorem~\ref{thm:main_banach}.
\end{proof}

Corollary~\ref{cor:main_banach} shows that quadratic invariance is
necessary and sufficient for convexity of $\hg$. We now show that
under additional invertibility assumptions, this result also applies
to the set of achievable closed-loop maps described in
Section~\ref{sec:intro}.

\begin{cor}\label{cor:clmap_banach}
Suppose the conditions of Corollary~\ref{cor:main_banach}
hold. Additionally, suppose $\mathcal{W}$ and $\mathcal{Z}$ are Banach
spaces, and $P_{11}\in\L(\mathcal{W},\mathcal{Z})$,
$P_{12}\in\L(\mathcal{U},\mathcal{Z})$, and
$P_{21}\in\L(\mathcal{W},\mathcal{Y})$. Finally, suppose $S \cap \mg =
S$.  If $P_{12}$ is left-invertible and $P_{21}$ is right-invertible,
then the following statements are equivalent.
\begin{enumerate}[(i)]
\item $S$ is quadratically invariant under $G$
\item $P_{11}-P_{12}\hg(S)P_{21} = \Gamma$, where $\Gamma$ is a convex set.
\end{enumerate}
\end{cor}

\begin{proof}
The proof of {\it(i)$\implies$(ii)} follows directly from
Theorem~\ref{thm:main_banach}. Conversely, suppose that {\it(ii)}
holds. Let $P_{12}^\dagger$ be a left-inverse of $P_{12}$ and let
$P_{21}^\dagger$ be a right-inverse of $P_{21}$. Then we have
\begin{align*}
\hg(S) &= P_{12}^\dagger P_{11} P_{21}^\dagger - P_{12}^\dagger \Gamma P_{21}^\dagger
\end{align*}
where the right-hand side is a convex set, because it is an affine
transformation of the convex set $\Gamma$. Applying
Corollary~\ref{cor:main_banach}, we conclude that $S$ is quadratically
invariant under $G$, as required.
\end{proof}

In Section~\ref{sec:examples}, we will show some examples that
illustrate why the invertibility requirements are necessary in this
result.

%\newpage

%%%%%%%%%%%%%%%%%%%%%%%%%%%%%%%%%%%%%%%%%%%%%%%%%%%%%%%%%%%%%%%%%%%%%%%%%%%%%%

\subsection{Extended space results}\label{sec:main_extended}

We now present results analogous to those of
Section~\ref{sec:main_banach}, but now for the extended
spaces~$\ell_{2e}$ and~$L_{2e}$.

\begin{thm}\label{thm:main_extended}
Suppose $S \subset \LTI(L_{2e}^{n_y},L_{2e}^{n_u})$ is inert, closed,
and homogeneous and $G \in \LTI(L_{2e}^{n_y},L_{2e}^{n_u})$. If
$\hg(S)=T$ where $T$ is a star-shaped set, then $T = S$. This result
also holds when $L_{2e}$ is replaced by $\ell_{2e}$.
\end{thm}

\begin{proof}
This proof is similar to that of Theorem~\ref{thm:main_banach}, except
we do not need to worry about the invertibility of $I-GK$, since it is
guaranteed by the inertness property of $S$. Suppose $\hg(S) = T$,
where $T$ is a star-shaped set. Applying $\hg$ to both sides, we
conclude that $\hg(T)=S$. Together with the star-shaped property of
$T$, it follows that $\hg(\alpha \hg(S)) \subset S$ for any $\alpha \in
[0,1]$. As in the proof of Theorem~\ref{thm:main_banach}, we conclude
that $g(\alpha) \in S$ for all $\alpha\in(0,1]$, where $g$ is defined
in~\eqref{def:g}. The rest of the proof follows as in the proof of
Theorem~\ref{thm:main_banach}. The only difference is that we must
verify continuity of $g$ at $0$ using the topology defined in
Section~\ref{sec:prelim}.
\end{proof}

As in Section~\ref{sec:main_banach}, Theorem~\ref{thm:main_extended}
may be specialized to the case where $S$ is a subspace, and combined
with Theorem~\ref{thm:qi_extended} to yield a necessary condition for
convexity of $\hg$.

\begin{cor}\label{cor:main_extended}
Suppose $S \subset \LTI(L_{2e}^{n_y},L_{2e}^{n_u})$ is an inert and
closed subspace and $G \in \LTI(L_{2e}^{n_y},L_{2e}^{n_u})$. Then the
following statements are equivalent
\begin{enumerate}[(i)]
\item $S$ is quadratically invariant under $G$
\item $\hg(S)=S$
\item $\hg(S)$ is convex
\end{enumerate}
This result also holds when $L_{2e}$ is replaced by $\ell_{2e}$.
\end{cor}

\begin{proof}
See the proof of Corollary~\ref{cor:main_banach}.
\end{proof}

Corollary~\ref{cor:main_extended} shows that quadratic invariance is
necessary and sufficient for convexity of $\hg$. As in
Section~\ref{sec:main_banach}, this result also applies to the set of
achievable closed-loop maps when we make additional invertibility
assumptions.

\begin{cor}\label{cor:clmap_extended}
Suppose the conditions of Corollary~\ref{cor:main_extended}
hold, and  $P\in\LTI(L_{2e},L_{2e})$.  If $P_{12}$ is
left-invertible and $P_{21}$ is right-invertible, then the following
statements are equivalent
\begin{enumerate}[(i)]
\item $S$ is quadratically invariant under $G$
\item $P_{11}-P_{12}\hg(S)P_{21}$ is convex
\end{enumerate}
This result also holds when $L_{2e}$ is replaced by $\ell_{2e}$.
\end{cor}

\begin{proof}
See the proof of Corollary~\ref{cor:clmap_banach}.
\end{proof}

%\newpage
%%%%%%%%%%%%%%%%%%%%%%%%%%%%%%%%%%%%%%%%%%%%%%%%%%%%%%%%%%%%%%%%%%%%%%%%%%%%%%%
\section{Examples}\label{sec:examples}

\subsubsection*{LQG with sparsity} 

Consider the extended space of signals $L_{2e}$, and suppose we have
the LTI map
\[
\bmat{P_{11} & P_{12} \\ P_{21} & G}:
	\bmat{L_{2e}^{n_w}\\L_{2e}^{n_u}} \to
	\bmat{L_{2e}^{n_z}\\L_{2e}^{n_y}}
\]
given by the following minimal state-space realization.
\[
\bmat{P_{11} & P_{12} \\ P_{21} & G} =
	\left[ \begin{array}{c|cc}
	A & B_1 & B_2 \\ \hlinet
	C_1 & 0 & D_{12} \\
	C_2 & D_{21} & 0
	\end{array}\right]
\]
We now make the classical assumptions typically made in $\Htwo$
synthesis. These assumptions ensure that the controller that achieves
a closed-loop map with minimum norm exists, is unique, and is
rational~\cite{zdg}. They are as follows.
\begin{enumerate}[(i)]\itemsep 0.2\baselineskip
\item $(C_2,A,B_2)$ is stabilizable and detectable 

\item 
For all $\omega\in\R$ the matrices $\bmat{ A-j\omega I & B_2 \\ C_1 &
D_{12}}$ and $D_{12}$ have full column rank

\item 
For all $\omega\in\R$ the matrices $\bmat{A-j\omega I & B_1 \\ C_2 &
D_{21}}$ and $D_{21}$ have full row rank

\end{enumerate}
Now suppose we seek a controller
$K \in \LTI(L_{2e}^{n_y},L_{2e}^{n_u})$ where $K$ has some prescribed
sparsity structure. So $K\in S$ for some closed subspace
$S$. Assumptions (ii) and (iii) ensure that $P_{12}$ is
left-invertible and $P_{21}$ is right-invertible,
respectively. Therefore, we may apply
Corollary~\ref{cor:clmap_extended} and conclude that quadratic
invariance is necessary and sufficient for convexity of the set of
achievable closed-loop maps. Thus for this class of problems,
quadratic invariance may be tested to determine convexity. In
particular, if $S$ is defined as the set of transfer functions with
desired sparsity or delays, then quadratic invariance may be
computationally tested using the methods in~\cite{rotkowitz06}.

It is worth noting that under the above assumptions there are only two
possibilities. Either $S$ is quadratically invariant, in which case
the set of closed-loop maps $C$ must be given by $C =
P_{11}-P_{12}SP_{21}$, or $S$ is not quadratically invariant in which
case $C$ is not convex. Additionally, one can conclude structural
properties. For example, if $S$ is a subspace, then for any plant $P$
satisfying these assumptions the only possible form of $C$ when it is
convex is that $C$ is affine; it cannot, for example, be a ball.

\subsubsection*{Non-affine example} 
In the results of Section~\ref{sec:main_banach}, as well as in the
first example of this section, the set of achievable closed loop maps
is affine whenever it is convex.  In this example, we show that
more complicated sets are achievable as well. Consider the Banach
space of real matrices under the standard induced 2-norm. Define the
matrices {\advance\arraycolsep-1pt
\begin{align*}
P_{11} &= \bmat{0\\0}
&
G &= \bmat{
0 & 1  & 0 & 0 \\ -1 & 0 & 1 & 0 \\ 0 & 0 & 0 & 1 \\ 1 & 0 & -1 & 0}
&
K &= \bmat{t&0&0&0\\0&t&0&0\\0&0&s&0\\0&0&0&s}
\end{align*}}%
Here, the subspace of admissible controllers $S$ is the set of all $K$
above parameterized by $(t,s)\in\R^2$. Now consider two possible pairs
of values for $P_{12}$ and $P_{21}$, and their corresponding sets of
achievable closed-loop maps $C = P_{11}-P_{12}\hg(K)P_{21}$.
\begin{enumerate}[(a)]
\item 
$P_{12} = \bmat{0 & 2 & 0 & 0 \\ 0 & 0 & 0 & -2}$ 
and 
$P_{21} = \bmat{0\\1\\0\\-1}$, 
\if\MODE3{}\else{\hfill\break}\fi
which leads to
$C = \bmat{ \frac{2s}{1+s^2+t^2} \\[1.5mm]
\frac{2t}{1+s^2+t^2}}$.

\item $P_{12} = \bmat{1 & 0 & 0 & 0 \\ 0 & 0 & -1 & 0}$
and 
$P_{21} = \bmat{0\\2\\0\\-1}$,
\if\MODE3{}\else{\hfill\break}\fi
which leads to
$C = \bmat{ \frac{(s^2+2)t^2}{1+s^2+t^2} \\[1.5mm]
\frac{s^2(1-t^2)}{1+s^2+t^2}}$.
\end{enumerate}
As we vary $(s,t)\in\R^2$, one may check that $C$ is the unit disc in
Case~(a), and a more complicated nonconvex shape in
Case~(b). Figure~\ref{fig:ex_shapes} shows the sets $C$ in each case.

\begin{figure}[ht]
\centering\begin{minipage}[b]{.45\textwidth}
	\centering
	\begin{tikzpicture}[twopic]
	\def\th{0.05}
	\draw[blfill] (0,0) circle (1);
	\draw[axes] (-1.5,0) -- (1.5,0);
	\draw[axes] (0,-1.4) -- (0,1.4);
	\draw[axes] (0,-1.4) -- (0,1.4);
	\draw[ticks] (1,-\th) -- (1,\th);
	\draw[ticks] (-1,-\th) -- (-1,\th);
	\draw[ticks] (-\th,1) -- (\th,1);
	\draw[ticks] (-\th,-1) -- (\th,-1);
	\node[anchor=south east] at (0,1) {$1$};
	\node[anchor=north east] at (0,-1) {$-1$};
	\node[anchor=north west] at (1,0) {$1$};
	\node[anchor=north east] at (-1,0) {$-1$};
	\end{tikzpicture}
	\subcaption{first case, $C$ is convex.}
\end{minipage}
\begin{minipage}[b]{.45\textwidth}
	\centering
	\begin{tikzpicture}[twopic]
	\def\xoff{.35}
	\def\poff{0.0}
	\def\th{0.05}
	\path[bfill] (0,0) -- (0,1) -- (2+\xoff,-1-\xoff) -- (2+\xoff,-\xoff) -- (2,0) -- (0,0);
	\draw[bline] (0,1) -- (0,0) -- (2,0);
	\draw[bdash,->] (0,1) -- (2+\xoff,-1-\xoff);
	\draw[bdash,->] (2,0) -- (2+\xoff,-\xoff);
	\draw[axes] (-.5-\poff,0) -- (2.5-\poff,0);
	\draw[axes] (0,-1.4) -- (0,1.4);
	\draw[ticks] (1,-\th) -- (1,\th);
	\draw[ticks] (2,-\th) -- (2,\th);
	\draw[ticks] (-\th,1) -- (\th,1);
	\draw[ticks] (-\th,-1) -- (\th,-1);
	\node[anchor=north east] at (1,0) {$1$};
	\node[anchor=south east] at (2,0) {$2$};
	\node[anchor=east] at (0,1) {$1$};
	\node[anchor=east] at (0,-1) {$-1$};
	\end{tikzpicture}
	\subcaption{second case, $C$ is nonconvex.}
\end{minipage} \\
\caption{\label{fig:ex_shapes}The set of achievable closed-loop
 maps described by both cases above. Solid lines are included in the
 set, while dashed lines are not.}
\end{figure}
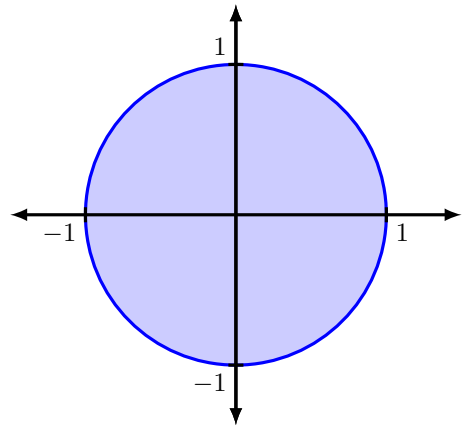
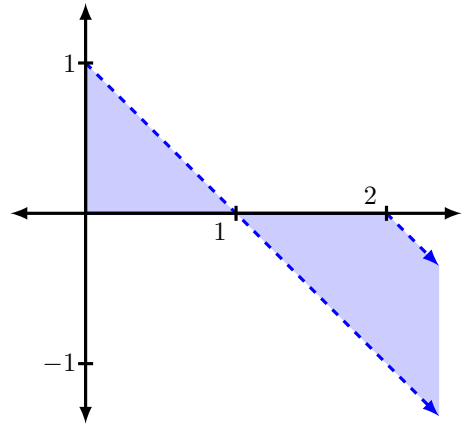

Both cases above have the same $G$ and $S$. It is straightforward to
check that $S$ is not quadratically invariant under $G$. Therefore,
Corollary~\ref{cor:main_banach} implies that $\hg(S)$ is not a convex
set. Unfortunately, we cannot apply Corollary~\ref{cor:clmap_banach}
to deduce anything about the convexity of $C$ because in both cases,
$P_{12}$ is not left-invertible and $P_{21}$ is not
right-invertible. Indeed, we have shown that both a convex set~(a) and
a nonconvex set~(b) are achievable depending on the choice of $P_{12}$
and $P_{21}$.

This example shows that in general, the convexity of $C$ does not
depend on $S$ and $G$ alone, but also on $P_{12}$ and $P_{21}$. This
idea is further discussed in~\cite{lessard2010iqi}, where $P_{12}$ and
$P_{21}$ are used as part of a sufficient condition for convexity that
is more general than quadratic invariance.

\newcommand{\B}{\rule[-1.2ex]{0pt}{0pt}}

\subsubsection*{Affine example} 
We now show an example where the set of closed-loop maps is a convex set, and in particular is affine, even though $S$ is not quadratically invariant under $G$. This example is a variant on an example shown in~\cite{lessard2009reduction}. Suppose
\begin{align*}
  P &= \left[\addtolength{\arraycolsep}{-1pt}\begin{array}{c;{1pt/2pt}ccc} \B a & b_1 & b_2 & b_2\\ \hdashline[1pt/2pt] \T
       c_1 & g_1 & 0 & 0 \\ \T
       c_1 &  g_1 & 0 & 0 \\ \T
        c_2 & g_2 & g_3 & g_3
       \end{array}\right] &
  S &= \left.\left\{ \addtolength{\arraycolsep}{-1pt}\bmat{\T k_1 & 0 & 0\\\T 0 & k_2 & 0\\\T 0 & 0 & k_3} \,\right\vert\, k_i \in \R \right\}
\end{align*}
Here $a$, $b_i$, $c_i$ and $g_i$ are real numbers for simplicity, although it is straightforward to construct similar examples over the rational transfer matrices. The matrix $P$ is partitioned into its four blocks as shown with $P_{22}=G$; the dashed lines do not denote a state-space representation. The information constraint $S$ is the subspace of controllers with a diagonal structure. Note that $KGK$ is not diagonal for all  diagonal $K$, so $S$ is not quadratically invariant under $G$. We will show that nonetheless that the set of achievable closed-loop maps $C$ is affine. To see why, define
\[
\t P = \left[\begin{array}{c;{1pt/2pt}cc} \B a & b_1 & b_2\\ 
    \hdashline[1pt/2pt] \T
    c_1 & g_1 & 0 \\ \T
    \T\B c_2 & g_2 & g_3
  \end{array}\right] \quad
\t S = \left.\left\{ 
    \bmat{k_1 & 0\\k_2 & k_3} \,\right\vert\, k_i \in \R \right\}
\]
Here $\t S$ is clearly quadratically invariant under $\t G$. Therefore the set of achievable closed-loop maps $\t C$ is affine. It is straightforward to check that $C =\t C$. In fact, the parameterizations are exactly equal provided the $k_i$ are the same. So $C$ is affine despite $S$ not being quadratically invariant under $G$. Such transformations $(P,S)\mapsto(\t P,\t S)$ are explored in~\cite{lessard2009reduction}.

%%%%%%%%%%%%%%%%%%%%%%%%%%%%%%%%%%%%%%%%%%%%%%%%%%%%%%%%%%%%%%%%%%%%%%%%%%%%%%%
\section{Conclusion}\label{sec:conclusion}

It was previously known that when the set of decentralized controllers
is a subspace~$S$, quadratic invariance is a necessary and sufficient
condition under which $\hg(S) = S$. In this paper, we showed that
quadratic invariance is in fact necessary and sufficient for $\hg(S)$
to be convex. Furthermore, we showed that under certain invertibility
conditions, quadratic invariance is also necessary and sufficient for
the convexity of~$C$, the set of achievable closed-loop maps.

This work therefore strengthens the utility of quadratic invariance as
an indicator for tractability of decentralized control synthesis
problems. However, there remains a nontrivial case, that when $S$ is
not quadratically invariant and the aforementioned invertibility
conditions do not hold. In that case one cannot draw any conclusions
about the convexity of~$C$. This was illustrated in
Section~\ref{sec:examples}, where we gave an example of a problem that
is not quadratically invariant, but~$C$ can be either convex or
nonconvex, depending on the choice of the system parameters.

%%%%%%%%%%%%%%%%%%%%%%%%%%%%%%%%%%%%%%%%%%%%%%%%%%%%%%%%%%%%%%%%%%%%%%%%%%%%%%%

\bibliographystyle{abbrv}
\bibliography{qico}

\end{document}